\documentclass[aps,superscriptaddress,twocolumn,floatfix,nofootinbib]{revtex4}

\usepackage{amsmath,amsthm,amsfonts,amssymb,times,bbm,graphicx,color}

\newtheorem{theorem}{Theorem}

\newtheorem{lemma}[theorem]{Lemma}

\newcommand{\bra}[1]{\mbox{$\langle #1 |$}}
\newcommand{\ket}[1]{\mbox{$| #1 \rangle$}}
\newcommand{\innerp}[2]{\mbox{$\langle #1 , #2 \rangle$}}
\newcommand{\ot}{\otimes}

\def\R{\mathbb{R}}

\def\N{\mathbb{N}}

\def\RR{\mathbbm{R}}

\def\A{\mathcal{A}}

\def\S{\mathcal{S}}
\def\K{\mathcal{K}}
\def\P{\mathcal{P}}

\def\Id{\mathbbm{1}}

\def\eins{\underline{1}}

\DeclareMathOperator{\tr}{tr}
\DeclareMathOperator{\Tr}{Tr}

\begin{document}

\title{All reversible dynamics in maximally non-local theories are trivial}

\author{David Gross}
\email{david.gross@itp.uni-hannover.de}
\affiliation{Institute for Theoretical Physics, Leibniz University Hannover, 30167 Hannover, Germany}
\author{Markus M\"uller}
\email{mueller@math.tu-berlin.de}
\affiliation{Institute of Mathematics, Technical University of Berlin, 10623 Berlin, Germany}
\affiliation{Institute of Physics and Astronomy, University of Potsdam, 14476 Potsdam, Germany}
\author{Roger Colbeck}
\email{colbeck@phys.ethz.ch}
\affiliation{Institute for Theoretical Physics, ETH Zurich, 8093 Zurich, Switzerland}
\affiliation{Institute of Theoretical Computer Science, ETH Zurich, 8092 Zurich, Switzerland}
\author{Oscar C. O. Dahlsten}
\email{dahlsten@phys.ethz.ch}
\affiliation{Institute for Theoretical Physics, ETH Zurich, 8093 Zurich, Switzerland}

\date{1st March 2010}

\begin{abstract} 
  A remarkable feature of quantum theory is non-locality (i.e.\ the
  presence of correlations which violate Bell inequalities).  However,
  quantum correlations are not maximally non-local, and it is natural
  to ask whether there are compelling reasons for rejecting theories
  in which stronger violations are possible. To shed light on this
  question, we consider post-quantum theories in which maximally
  non-local states (non-local boxes) occur.  It has previously been
  conjectured that the set of dynamical transformations possible in
  such theories is severely limited. We settle the question
  affirmatively in the case of reversible dynamics, by completely
  characterizing all such transformations allowed in this setting. We
  find that the dynamical group is trivial, in the sense that it is
  generated solely by local operations and permutations of systems.
  In particular, no correlations can ever be created; non-local boxes
  cannot be prepared from product states (in other words, no analogues
  of entangling unitary operations exist), and classical computers can
  efficiently simulate all such processes.
\end{abstract}

\maketitle

\emph{Introduction.}|Quantum mechanics exhibits the remarkable
feature of non-local correlations, as highlighted in Bell's seminal
paper~\cite{Bell87}.  Such correlations have (up to a few remaining
loopholes) been extensively verified in experiments~\cite{Aspect99}.

Aside from their theoretical importance, non-local correlations can be
exploited for technological use:
they are vital in entanglement-based quantum key distribution
schemes~\cite{Ekert91}, for example, where their presence can be used
to guarantee security (see also~\cite{Scarani08} for a recent review).

While quantum mechanics violates Bell inequalities, it does not do so
in the maximal possible way.  There are conceivable devices, so-called
\emph{non-local} or \emph{Popescu-Rohrlich boxes}, that permit even
stronger correlations than quantum mechanics does, while respecting
the no-signalling principle~\cite{Tsirelson,Tsirelson2,PopescuR94}.
Such correlations are not observed in nature and the question arises
as to whether other fundamental principles might be violated if they
were to exist.

There has already been some progress towards answering this
question. For example, the existence of non-local boxes would
lead to some communication complexity problems 
becoming trivial~\cite{vanDam05,BrassardBLMTU06}, the possibility of
oblivious transfer~\cite{BCUWW} and the lack of so-called information
causality~\cite{PPKSWZ}.  It has also been realized that in a theory
in which maximally Bell violating correlations emerge, the set of
possible dynamical transformations would be severely restricted
compared to those allowed in quantum theory~\cite{Barrett07}.
While a complete classification of the dynamics has remained elusive,
it has been shown, for example, that entanglement swapping is
impossible~\cite{ShortPG05,ShortBarrett09}.  Furthermore, the question
of the computational power of such a theory has been
raised~\cite{Barrett07,ShortBarrett09}.

\begin{figure}
\includegraphics[width=0.3\textwidth]{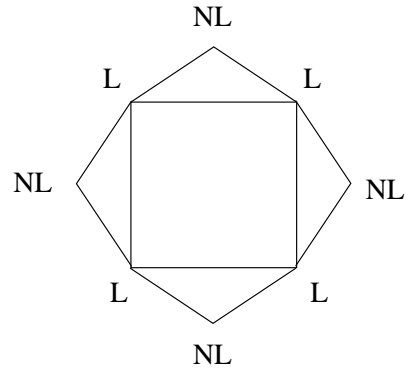}
\caption{Two-dimensional caricature of the (normalized) boxworld state
  space formed by stellating a square.  Local vertices are denoted by
  L and non-local ones by NL.  No symmetries of this object take L
  states to NL states or vice versa.}
\label{fig:1}
\end{figure}

We work in the framework of generalized probabilistic
theories~\cite{BarnumBLW07, Hardy01, Barrett07, Holevo}, adopting the
pragmatic operational view that the physical content of a theory is in
the predicted statistics of measurement outcomes given preparations
and transformations.  The framework makes minimal assumptions and
allows for mathematical rigour.  We consider a system composed of $N$
subsystems.  To each subsystem one of $M\geq 1$ measurements may be
applied, yielding one of $K\geq 2$ outcomes (in the following, unless
otherwise stated, we assume each subsystem has the same $M$ and $K$).
The state space contains all non-signalling correlations,
corresponding to so-called \emph{generalized non-signalling
  theory}~\cite{Barrett07} or, more colloquially, {\em boxworld}.

Our main result (Theorem 1) is that (except in the case $M=1$ which
corresponds to classical theory) the set of reversible transformations
in boxworld is trivial: \emph{all} such operations are a combination
of local operations on a single system (which correspond to
relabellings of measurements and their outcomes) and permutations of
local systems (which correspond to relabellings of subsystems).  This
solves the aforementioned open problem concerning the computational
power of boxworld in the case of reversible
dynamics~\cite{Barrett07,ShortBarrett09}.

Another interesting consequence is that, in boxworld, measurements and
dynamics are necessarily distinct physical processes, in the sense
that a measurement cannot be seen as a reversible dynamics on the
system comprising the state and measurement device (cf.\ quantum
theory, where the measurement process can be seen as a unitary
evolution from the point of view of an external observer).  We discuss
this further in the final section.

We note that, in the case of a classical-boxworld hybrid system,
Theorem~1 does not hold|we give an example of a CNOT operation on this
system at the end of the paper.  However, for all types of system,
including those where the number of measurements and outcomes differs
among the subsystems, reversible dynamics map pure product states to
pure product states|that is, non-local states cannot be reversibly prepared from
product states.  This is our second main result (Theorem~2).

A geometric intuition behind this result is as follows.  The state
space of the theory is a convex polytope, and reversible
transformations must map it to itself.  They therefore correspond to
symmetries of the polytope.  The polytope is in some way
\emph{stellated}, with the vertices corresponding to maximally
non-local states having a different character from local ones.  They
are hence not connected by symmetries of the polytope.  A
two-dimensional caricature is shown in Figure~\ref{fig:1}.

The presentation proceeds as follows. We begin by formally
introducing boxworld, then proceed to give the mathematical framework
we will work with. This is essentially the standard generalized
probabilistic framework, as used in~\cite{BarnumBLW07, Hardy01,
  Barrett07, Holevo}.  For clarity of exposition, in the main text we
restrict to the case of two binary measurements ($M=2$, $K=2$) and
give proofs of the main theorems for this case.  The general case is
deferred to the appendix, where the proofs are slightly more
complicated but analogous.
\bigskip

\emph{Boxworld.}|Recall that we have a system comprising $N$
subsystems and, on each subsystem, one of $M$ possible measurements
can be applied (corresponding to different measurement devices),
yielding one of $K$ possible outcomes (in the most general case, $K$
depends on the measurement).  The local measurements are denoted
$\{X_0,X_1,\ldots,X_{M-1}\}$.  A measurement on the entire system made up
of local measurements can then be described by a string $A_1\dots A_N$,
where $A_i\in\{X_0,X_1,\ldots,X_{M-1}\}$ specifies the measurement applied
to the $i$th subsystem.  Similarly, the corresponding outcomes are
denoted $a_1\dots a_N$, with $a_i\in\{0,1,\ldots,K-1\}$.
Measurement-outcome pairs are called \emph{effects}, e.g.\ a
measurement of $X_1$ giving outcome $3$.  A \emph{state} is then a
function $P: (a_1\dots a_N | A_1\dots A_N)\mapsto[0,1]$, which gives
the probability of the effect that $A_1\dots A_N$ is measured and
gives outcomes $a_1\dots a_N$.  More general measurements are
possible: a measurement is a collection of effects for which the sum
of the outcome probabilities over the collection is 1 when acting on
any state.  Such measurements include procedures whereby the
measurement performed on a particular subsystem depends on the outcomes
of previous measurements, convex combinations of such procedures and
more~\cite{ShortBarrett09}.  However, the statistics of the local
measurements $A_1\dots A_N$ are sufficient to uniquely determine the
outcome probabilities of all measurements, and hence can be used to
specify the state.  This non-trivial assumption is known as the local
observability principle~\cite{DAriano07}.

Furthermore, the subsystems can be spatially separated, and hence we
require that $P$ satisfies the \emph{non-signalling} conditions, i.e.\
that
\begin{eqnarray}\label{eqn:nos}
\sum_{a_i=0}^{K-1} P(a_1,\ldots,a_i,\ldots,a_N|A_1,\ldots,A_i,\ldots,A_N)
\end{eqnarray}
is independent of $A_i$.  This implies that the marginal
distribution on some set of subsystems is independent of the choice of
measurement(s) on other subsystems.

Boxworld is a physical theory whose state space consists of any $P$
subject to: {(\it i)} $P$ takes values in $[0,1]$; {\it (ii)} $P$ is
normalized in the obvious sense; and {\it (iii)} $P$ satisfies the
non-signalling conditions (\ref{eqn:nos}).  The constraints {\it (i)}
-- {\it (iii)} are such that the state space is a convex polytope,
which turns out to have a non-trivial structure.

We first deal with the special case $M=K=2$ (the case of so-called
\emph{gbits}~\cite{Barrett07}).  The corresponding state spaces
(defined below) contain interesting non-local states, for example,
non-local boxes with maximal Bell violating correlations.  We label
the measurements $X_0=X$ and $X_1=Z$.\bigskip

\bigskip {\it Mathematical Framework.}|We work in the
\emph{generalized probabilistic framework} (see e.g.\
\cite{BarnumBLW07, Hardy01, Barrett07, Holevo}).  Here, states are
represented as vectors embedded in a real vector space.  Effects will
also be represented as vectors, such that the probabilities of
outcomes will be given by inner products between the relevant vectors.
We begin with the case of a single system ($N=1$). We choose three
linearly independent vectors $X, Z, \eins\in \RR^3$.  The vector $X$
is identified with $(1|X)$, which is the effect that the $X$
measurement gives outcome $1$.  We define a vector $\neg X := \eins -
X$ and associate it with $(0|X)$. The prefix $\neg$ may be interpreted
as a negation.  Lastly, the $\neg Z$ effect is defined analogously as
$\neg Z:=\eins-Z$.  Because $X, Z, \eins$ are linearly independent,
for every state $P$, there is a unique vector $s\in\RR^3$ representing
$P$ in the sense that
\begin{equation*}
	\langle X, s \rangle = P(1|X),\quad
	\langle Z, s \rangle = P(1|Z),\quad
	\langle \eins, s \rangle = 1.
\end{equation*}
It follows that $\langle \neg X, s \rangle = P(0|X)$ and likewise for
$Z$. We will refer to the set $\P^{(1)}=\{X, \neg X, Z, \neg Z\}$ as
the single-site \emph{extremal effects}, for reasons that will become
clear below. (Note that the quantum analogue of our effect vectors are
projectors, and the inner product is analogous to the Hilbert-Schmidt
scalar product, mapping states, $\rho$, and projectors, $\Pi$, to
probabilities, $\Tr(\rho \Pi)$.)

The $N$-subsystem extremal effects $\P^{(N)}$ are defined to be the
tensor products $A_1\otimes\dots\otimes A_N$, where $A_i\in\P^{(1)}$
(the reason for this definition is that it recovers the full set of
non-signalling distributions for the state space, as will be shown in
Lemma 1).  We further define the identity on $N$ sites,
$\eins^{(N)}:=\eins\otimes\ldots\otimes\eins$.  A central object is
the convex cone $\K^{(N)}$ generated by $\P^{(N)}$~\cite{Aliprantis}.
This cone is the collection of all vectors which can be
written as a linear combination of elements of $\P^{(N)}$ with
non-negative coefficients.  For any convex cone $\K$, one can
associate a \emph{dual cone} $\K^*=\{s\,|\,\langle A, s\rangle \geq
0\,\forall\, A\in\K\}$.  We will identify this with the set of
unnormalized states.

Our interest in cones and duality stems from the following lemma, which
characterizes the state space of boxworld in terms of the cone
$\K^{(N)}$.  It also implies the well-known result that there are no
entangled effects in boxworld.

\emph{Lemma 1.}|Let $\S^{(N)}$ be the set of vectors $s$ in the dual
cone $(\K^{(N)})^{*}$ which satisfy $\langle \eins^{(N)},s \rangle =
1$.  The space of (normalized) states in boxworld can be
represented by $\S^{(N)}$.

\begin{proof}
  We use the notation $\neg^0 A:=A$ and $\neg^1 A:=\neg A$ for
  $A\in 
  \{X,Z\}$.  The vectors $s\in \S^{(N)}$ will henceforth be called
  \emph{states}; they satisfy $\langle \eins^{(N)},s\rangle=1$ and
  $\langle B,s\rangle\geq 0$ for all $B\in \K^{(N)}$.  To every state
  $s$, we associate a probability distribution $P$ via
\[
   P(a_1,\ldots,a_N|A_1,\ldots,A_N):=\langle \neg^{a_1} A_1\otimes\ldots \otimes \neg^{a_N} A_N,s\rangle
\]
for $A_i\in\{X,Z\}$ and $a_i\in\{0,1\}$. First we show that every such
$P$ is a valid non-signalling probability distribution. By definition,
$P$ is non-negative. To see that it is normalized, note that we can
decompose the identity
\begin{equation}\label{eq:id1}
   \eins^{(N)}=\sum_{x\in\{0,1\}^N} \neg^{x_1} A_1 \otimes \neg^{x_2} A_2 \otimes \ldots \otimes \neg^{x_N} A_N
\end{equation}
for any choices of $A_i\in\{X,Z\}$, so that
\[
\sum_{a_1,\ldots,a_N} P(a_1,\ldots,
a_N|A_1,\ldots,A_N)=\innerp{\eins^{(N)}}{s}=1.
\]
To see that $P$ is non-signalling, consider
\begin{eqnarray*}
&&\sum_{a_i}P(a_1,\ldots,a_i,\ldots,a_N|A_1,\ldots,A_N) \\
&=&\innerp{\neg^{a_1}A_1\otimes\ldots\otimes A_i \otimes \ldots \otimes \neg^{a_N} A_N}{s}\\
&&+\innerp{\neg^{a_1}A_1\otimes\ldots\otimes \neg A_i \otimes \ldots \otimes \neg^{a_N} A_N}{s} \\
&=&\innerp{\neg^{a_1}A_1\otimes\ldots\otimes \eins \otimes \ldots \otimes \neg^{a_N} A_N}{s}
\end{eqnarray*}
which is independent of $A_i$.

To show that every non-signalling distribution has an associated
state, note that there is a unique vector $s\in(\R^3)^{\ot N}$ such
that $\innerp{\neg^{a_1}A_1\ot\ldots\ot
  \neg^{a_N}A_N}{s}=P(a_1,\ldots,a_N|A_1,\ldots,A_N)$ for
$\neg^{a_i}A_i\in\{X,Z,\neg X\}$ (since these effects span the
space). It is then easy to see that the no-signalling property
enforces consistency also in the case that $\neg^{a_i} A_i=\neg Z$ for
some $i$.  Non-negativity and normalization follow directly from the
corresponding statements for the probability distribution $P$.
\end{proof}

\bigskip

\emph{Transformations.}|We now consider transformations in boxworld.
First note that all allowed dynamical transformations in general
probabilistic theories (reversible or not) are linear|this follows
from the fact that they have to respect convex combinations, which
correspond to probabilistic mixtures.
For a general proof of this fact see~\cite{Barrett07}.

The {\em allowed} transformations, $T$, are defined to be linear maps
with the property that for all $s\in \S^{(N)}$, $Ts\in \S^{(N)}$.  A
transformation is {\em reversible} if both $T$ and $T^{-1}$ are
allowed transformations.  It follows that a reversible transformation
maps the state space $\S^{(N)}$ bijectively onto itself.  Furthermore,
since $T$ is a linear map, it is also the case that $T$ maps extremal
states to extremal states.  (More generally, one would only consider a
transformation allowed if $T\ot\Id$ is also allowed|a condition
analogous to complete positivity in quantum theory.  However, our
result applies without this additional requirement.)

Note that the states $s\in\S^{(N)}$ themselves do not have a physical
meaning|only their scalar products with effects do, i.e.\ $\langle A,
s \rangle$ (which are probabilities).  Since $\langle A,Ts\rangle =
\langle T^\dagger A,s \rangle$, the dynamics may equivalently be
specified by means of the adjoint map $T^\dagger$.  (In quantum
theory, the analogue is passing from the Schr\"odinger to the
Heisenberg picture.)  Reversible transformations, $T$, map the state
space bijectively onto itself and, likewise, the adjoint
transformations $T^\dagger$ act accordingly on the cone of effects
$\K^{(N)}$.

\emph{Lemma 2.}|Adjoint reversible transformations $T^\dagger$ map the
cone of effects $\K^{(N)}$ bijectively onto itself. Moreover, they map
the set of {\em extremal effects}, $\P^{(N)}$, onto itself.
\begin{proof}
  Any vector $t\in \K^{(N)*}$ can be written $t=\lambda s$ for some
  $\lambda\geq 0$ and some $s\in\S^{(N)}$.
  Then, for any outcome $B\in \K^{(N)}$, we have
  \[
  \langle T^\dagger B, t\rangle=\langle B,T t\rangle= \lambda\langle
  B,Ts\rangle\geq 0,
  \]
  since $Ts\in\S^{(N)}$. From the definition of the dual cone, it
  follows that $T^\dagger B\in \left(\K^{(N)}\right)^{**}=\K^{(N)}$
  (note that $\K^{**}=\K$ for every closed convex cone $\K$
  (cf.~\cite{Aliprantis})).  Therefore, $T^\dagger$ maps the cone
  of effects $\K^{(N)}$ into itself.  The same argument applies to
  the inverse $(T^\dagger)^{-1}=(T^{-1})^\dagger$, hence the cone is
  mapped bijectively onto itself.

  Since it is a convex cone, $\K^{(N)}$ is completely characterized by
  its extremal rays.  By linearity, $T^\dagger$ maps the extremal rays
  of $\K^{(N)}$ onto themselves.  From the definition of $\K^{(N)}$,
  we know that the cone is the convex hull of the $4^N$ rays formed by
  all $A\in\P^{(N)}$.  It is elementary to check that these are indeed
  the extremal rays. Therefore, for every $A\in\P^{(N)}$, there exists
  an $A'\in\P^{(N)}$ and a non-negative number $\lambda$ such that
  $T^\dagger(A)=\lambda A'$. To see that $\lambda$ must equal $1$,
  observe that for every $B\in\P^{(N)}$, there exist (product) states
  $s_0, s_1\in\S^{(N)}$ such that $\langle s_0, B\rangle = 0$ and
  $\langle s_1, B \rangle=1$. Since this holds in particular for both
  $A$ and $A'$, it follows that $\lambda=1$.
\end{proof}

\emph{Orthogonal representation of transformations.}|There are $4^N$
extremal effects, and thus $4^N!$ permutations acting on $\P^{(N)}$.
We go on to show that only a tiny fraction of those is actually
realizable in boxworld.  It will be convenient to use a specific
representation of $X$, $Z$ and $\eins$:
\begin{eqnarray*}
	X &=& (1/2,1/\sqrt 2,0), \ \ 
	Z = (1/2,0,1/\sqrt 2), \ \
	\eins = (1,0,0).
\end{eqnarray*}

\emph{Lemma 3.}|With respect to the representation above, it holds
that any reversible transformation $T$ is orthogonal, i.e.\ on $N$
subsystems, $T^\dagger T=\Id_{3^N}$, where $\Id_d$ is the
$d$-dimensional identity matrix.
\begin{proof}
  First observe that with this choice,
  $\sum_{A\in\P^{(1)}}\ket{A}\bra{A}=\Id_3$ and hence (since
  $\P^{(N)}$ factorizes)
  $\sum_{A\in\P^{(N)}}\ket{A}\bra{A}=\Id_{3^N}$.  Then, since $T^\dagger$
  permutes the extremal effects,
  $T^\dagger T=T^\dagger \left(\sum_{A\in\P^{(N)}}\ket{A}\bra{A}\right)T
  =\sum_{A\in\P^{(N)}}\ket{A}\bra{A}=\Id_{3^N}$.
\end{proof}

The fact that $T$ (and thus $T^\dagger$) is orthogonal, gives rise to
a host of invariants.  If one picks any two extremal effects $Q,
R\in\P^{(N)}$, then clearly their inner product is a conserved
quantity: $\innerp QR = \innerp{T^\dagger Q}{T^\dagger R}$. However,
$|\innerp{Q}{R}|=4^{-N} 3^{N-d_H(Q,R)}$, where $d_H(Q,R)$ is the
\emph{Hamming distance} between $Q$ and $R$, i.e.\ the number of
places at which $Q$ and $R$ differ. Thus the Hamming distance of
extremal effects is a conserved quantity: $d_H(Q,R)=d_H(T^\dagger Q,
T^\dagger R)$.

It is well-known in the theory of error
correction~\cite{HammingDistance} that the set of maps on finite
strings which preserve the Hamming distance is highly restricted: the
group of those maps is generated by local transformations and
permutations of sites only (for a proof, see the Appendix).
Thus $T^\dagger$ acts as such an operation on $\P^{(N)}$.  Moreover,
since the states in $\P^{(N)}$ span the entire space, the action on
this set is sufficient to completely specify $T^\dagger$.

Furthermore, it is straightforward to show that the set of allowed
local operations comprises exchanging $X$ and $Z$ (relabelling
measurements), exchanging $X$ and $\neg X$ (relabelling the outcome
upon input $X$), exchanging $Z$ and $\neg Z$ (relabelling the outcome
upon input $Z$) and combinations thereof (see Lemma~\ref{lem:9} in the
Appendix, for a proof in the general case).

\bigskip

\emph{Main results.}|Combining all the previous results proves the
following theorem in the special case of $M=2$ measurements with $K=2$
outcomes (a full proof for all $M\geq 2$ and $K\geq 2$ is given in the
appendix):

\bigskip

\emph{Theorem 1.}|Every reversible transformation on a system comprising
$N$ subsystems in boxworld, with $M\geq 2$ measurements at every subsystem
each having $K\geq 2$ outcomes, is a permutation of subsystems, followed
by local relabellings of measurements and their outcomes.\bigskip

Furthermore, we show the following:
\bigskip

\emph{Theorem 2.}|In boxworld, every reversible transformation maps
pure product states to pure product states. This is true even if the
system is coupled to an arbitrary number of classical systems, and if
the number of devices and outcomes varies from subsystem to subsystem.
\bigskip

Before giving the proof, we need to slightly extend the notion of
outcome vectors to the general case.  We denote the set of extremal
effects for the $i$th subsystem by $\P^i=\{X^i_m(k)\}$, where $m$
labels the measurements (the number of different $m$s may depend on
$i$) and $k$ the corresponding outcomes (the number of different $k$s
may depend on $m$ and on $i$).  These vectors satisfy $\sum_k
X^i_m(k)=\eins^i$, where $\eins^i$ represents the identity. Except for
these relations, no linear dependencies occur.

The identity on the full system is then
$\eins^{(N)}:=\eins^1\otimes\ldots\otimes \eins^N$, and the extremal
effects are $\P^{(N)}:=\P^1\otimes\ldots\otimes\P^N$. The convex cone
$\K^{(N)}$ and the state space $\S^{(N)}$ are defined analogously to
the binary case previously described. The statements and proofs of
Lemmas~1 and~2 remain valid in this more general case, hence, in
particular, adjoint reversible transformations map $\P^{(N)}$ onto
itself.

\begin{proof}
  To complete the proof of Theorem~2, note that a state $s\in\S^{(N)}$
  is a pure product state (that is, of the form
  $s=s_1\otimes\ldots\otimes s_N$, where all $s_i$ are pure) if and
  only if $\langle A,s\rangle\in\{0,1\}$ for all extremal effects
  $A\in\P^{(N)}$ (a proof is given in
  Lemma~\ref{lem:extremal01}). Suppose that $s$ is a pure product
  state and $T$ a reversible transformation, then
\[
   \langle A,Ts\rangle=\langle T^\dagger A,s\rangle\in\{0,1\}\quad\mbox{for all }A\in\P^{(N)},
\]
which proves that $Ts$ must also be a pure product state.
\end{proof}

Note that Theorem~1 does not, in general, apply to the case of
site-dependent numbers of measurements.  For example, suppose that we have
two sites, where the first has two binary measurements, $X$ and $Z$,
and the second allows only a single binary measurement, $Y$. (In other
words, a gbit is coupled to a classical bit.)  It is then
straightforward to construct a reversible CNOT operation, where the
classical bit is the control bit. For example, there is an adjoint
reversible transformation that acts as
\[
   A\otimes Y \mapsto A \otimes Y,\qquad A\otimes \neg Y \mapsto \neg A \otimes \neg Y
\]
for all $A\in\{X,Z,\neg X,\neg Z\}$. 

In the case of a system composed of several classical subsystems,
Theorem~1 also does not hold|the dynamics in such a case is
non-trivial.  Nevertheless, Theorem~2 does apply to this case|it
remains impossible to prepare entangled states from separable ones.

\bigskip

\emph{Conclusions.}|We have shown that the set of reversible
operations in boxworld is trivial: the only possible operations
relabel subsystems, local measurements and their outcomes.  In
particular, there is no boxworld analogue of an entangling unitary in
quantum theory, one cannot reversibly prepare non-local states from
separable ones, nor perform useful computations reversibly.

In addition, the results have consequences for the interplay between
dynamics and measurements in boxworld: suppose we have a system
comprising a particle, $A$, and two observers, $B$ and $C$, initially
in an uncorrelated tripartite product state.  In quantum theory, if
$B$ measures $A$, but $C$ does not take part in the interaction, then
$C$ can model the corresponding dynamics by a unitary transformation
on the $AB$-system.  That is, $C$ can view the whole interaction as
reversible while retaining the ability to correctly predict the
outcome probabilities of any future measurements.  (Theories with such
a property might be called \emph{fundamentally reversible}.)
In boxworld, on the other hand, this is not true: $B$'s measurement on
$A$ would have to create correlations between $A$ and $B$, but this
could never be achieved by a reversible transformation.  Hence $C$
would have to model the $AB$-measurement using irreversible dynamics,
even if $C$ did not take part in the interaction itself.

It would be interesting to extend our result to explore which state
spaces are compatible with fundamentally reversible theories in this
sense, or with theories that are \emph{transitive}, i.e.\ that every
pure state can be reversibly mapped to any other. This property has
been used by Hardy as an axiom for quantum theory~\cite{Hardy01}.
Both conditions seem to strongly restrict the possible geometry of the
state space, and an interesting open question is how non-local such
theories can be.

\bigskip

\emph{Acknowledgments.}|We are grateful to Howard Barnum, Volkher
Scholz, and Reinhard Werner for interesting discussions and to Nicholas
Harrigan for comments which improved the presentation.  OD would like
to thank Matthew Leifer for introducing him to generalized
probabilistic theories, and MM would like to thank Christopher Witte
for the same reason and for the quadratic cocktail model of a gbit.
RC and OD acknowledge support from the Swiss National Science
Foundation (grant No. 200021-119868). DG's research is supported by
the EU (CORNER).

\section*{Appendix}
This appendix contains the proof of Theorem~1: \emph{Every reversible
  transformation on a system comprising $N$ subsystems in boxworld,
  with $M\geq 2$ measurements at every subsystem each having $K\geq 2$
  outcomes, is a permutation of subsystems, followed by local
  operations.}

The proof idea is the same as in the case $M=K=2$: find a particular
representation of the vectors $X_m(k)$ and $\eins$ (corresponding to
the previous vectors $X$, $Z$, $\neg X$ and $\neg Z$ and $\eins$) such
that reversible transformations are orthogonal, and such that the
scalar products of those vectors yield useful invariants.  We recall
that the dual of reversible transformations preserve the cone of
effects and so permute extremal effects (this is Lemma~2 applied to
this case).

We start with the following observation:
\setcounter{theorem}{3}
\begin{lemma}For every $N\in\N$, there exist unit vectors
  $\{w_i\}_{i=1}^{N+1}$ in $\R^N$ with the properties
\begin{itemize}
\item $\langle w_i,w_j\rangle=-\frac 1 N$ if $i\neq j$,
\item $\displaystyle \sum_{i=1}^{N+1}w_i=0$, and
\item $\displaystyle \frac 1 {N+1} \sum_{i=1}^{N+1}|w_i\rangle\langle w_i|=\frac 1 N \Id_N$.
\end{itemize}
\end{lemma}
\begin{proof}
  Rather than giving the vectors explicitly, we construct them
  implicitly from the standard $N$-simplex in $\R^{N+1}$: let $e_i$ be
  the $i$th standard unit vector in $\R^{N+1}$, and $c$ the center of
  those vectors, that is $c:=\frac 1 {N+1}\sum_{i=1}^{N+1}e_i$.
  Define $v_i:=e_i-c$, so that the angles between those vectors
  ($i\neq j$) are $\frac{\langle v_i,v_j\rangle}{\|v_i\|\,\|v_j\|} =
  -\frac 1 N$.  By construction, we have $\sum_{i=1}^{N+1}v_i=0$, so
  the vectors are linearly dependent.  The set $\{w_i\}_{i=1}^{N+1}$
  are then the vectors resulting from embedding normalized versions of
  the $v_i$s isometrically into $\R^N$.  The first two claimed
  equalities follow immediately.  The third can be confirmed by
  computing\footnote{The Schatten 2-norm is defined
    by $\|A\|^2_2:=\tr(AA^\dagger)$.}
\[
   \left\| \frac 1 {N+1}\sum_{i=1}^{N+1}\ket{w_i}\bra{w_i}-\frac 1 N \Id_N\right\|_2^2=0,
\]
which involves only scalar products of the form $\langle w_i,w_j\rangle$.
\end{proof}

The vectors $X_m(k)$ representing the $k$th measurement outcomes for
the $m$th measurement (counting from zero) can be constructed as follows:
\begin{itemize}
\item Choose $\eins\neq0$ arbitrarily,
\item choose $X_m(0),X_m(1),\ldots,X_m(K-2)$ for all
  $m$ such that all obtained vectors are linearly independent,
\item define $X_m(K-1)$ as
  $\eins-\sum_{k=0}^{K-2}X_m(k)$.
\end{itemize}
We choose these in a particular way in order to simplify the
subsequent argument:  the single-site effects will be vectors in
$\R^{M(K-1)+1}=(\R^M\otimes \R^{K-1})\oplus \R$.  Let $e_m$
denote the $m$th standard unit vector in $\R^M$, and $\eins$ be
the unit vector on the direct sum space, $\R$. Then, define
\begin{equation}
   X_m(k):=\sqrt{\frac{M(K-1)}{K^2}} e_{m+1}\otimes w_{k+1}+\frac 1 K \eins
   \label{eq:repXm}
\end{equation}
for $0\leq m\leq M-1$ and $0\leq k \leq K-1$.
Useful properties of these vectors are given in the following lemma.

\begin{lemma}In the representation given above, we have
\begin{eqnarray}
   \sum_{k=0}^{K-1} X_m(k)&=&\eins,\label{eq:decomp}\\
   \sum_{m=0}^{M-1} \sum_{k=0}^{K-1} \ket{X_m(k)}\bra{X_m(k)} &=&
   \frac M K \Id,\ \text{ and}
\end{eqnarray}
$$\langle X_m(k),X_{m'}(k')\rangle=\frac 1 {K^2}\left\{
      \begin{array}{cl}
         1 & m\neq m' \\
         1-M & m=m',k\neq k'\\
         1+M(K-1) & m=m',k=k'.
      \end{array}
   \right.
$$
Moreover, reversible transformations are orthogonal with respect to
this representation.
\end{lemma}

\begin{proof}
  The three equations can be verified by direct calculation.  That
  reversible transformations are orthogonal is simply the extension of
  Lemma~3 to the present case.
\end{proof}

We remark that the inner product $1-M$ for $m=m',k\neq k'$ is the
reason why Theorem 2 does not hold in the case of classical systems
($M=1$).  
In the following, we will assume that $M\geq 2$.

We now consider an adjoint reversible transformation, $T^\dagger$.
Note that for all $s\in\S^{(N)}$,
$$1=\innerp{\eins^{(N)}}{Ts}=\innerp{T^\dagger\eins^{(N)}}{s},$$ from which it
follows that $T^\dagger\eins^{(N)}=\eins^{(N)}$.  Moreover, we have
the following property:

\begin{lemma}\label{lem:6}Let $Q,R\in\P^{(N)}$ be two extremal effects that
  differ at exactly one site, and let $T^\dagger$ be an adjoint
  reversible transformation. Then, $T^\dagger Q$ and $T^\dagger R$
  also differ at exactly one site.
\end{lemma}
\begin{proof}
  Since $Q$ and $R$ factorize, we can compute the inner product
  $\langle Q,R\rangle$ termwise. Let $i$ be the site where $Q$ and $R$
  differ, and let $Q_i$ and $R_i$ be the corresponding factors.  

  First, consider the case that $Q_i$ and $R_i$ represent different
  outcomes of the \emph{same} measurement. Then, the inner product is
  the negative value
  \[
  K^{2N}\langle Q,R\rangle=(1-M)\left(1+M(K-1)\right)^{N-1}
  \]
  which is the smallest value that can possibly be attained. Hence
  $K^{2N}\langle T^\dagger Q,T^\dagger R\rangle$ has the same value,
  such that $T^\dagger Q$ and $T^\dagger R$ also differ at a single
  site only (where they refer to different outcomes of the same
  measurement).

  The alternative case is where $Q_i$ and $R_i$ represent outcomes of
  \emph{different} measurements.  Note that $\eins-Q_i-R_i\notin\K$
  and hence $\eins^{(N)}-Q-R\notin\K^{(N)}$.  Since $T^\dagger$ preserves
  $\eins^{(N)}$ and maps the cone bijectively to itself (cf.\ Lemma~2), we
  have $\eins^{(N)}-T^\dagger Q-T^\dagger R\notin\K^{(N)}$, from which it
  follows that $T^\dagger Q$ and $T^\dagger R$ correspond to outcomes
  of different measurements on at least one factor.  Furthermore,
  $$K^{2N}\langle Q,R\rangle=\left(1+M(K-1)\right)^{N-1}$$ is preserved.
  This is the largest value that can be attained subject to the
  constraint that they represent outcomes of different measurements on
  at least one factor.  It follows that $T^\dagger Q$ and $T^\dagger
  R$ are identical in all but one tensor factor.
\end{proof}

The proof of Theorem 1 is now completed using some properties of the
Hamming distance.  The list of local effects (listing the measurement-outcome
pairs at the successive sites) can be used to form a string in
$\mathbb{Z}_d^N$, where $d=MK$.  The Hamming distance between two
strings, $Q$ and $R$, is defined by $$d_H(Q,R):=|\{i:Q_i\neq R_i\}|.$$
Lemma~\ref{lem:6} shows that if $Q,R\in\P^{(N)}$ are two arbitrary
extremal effects with $d_H(Q,R)=1$, then the transformed effects
satisfy $d_H(T^\dagger Q,T^\dagger R)=1$.

In fact, all reversible operations that preserve Hamming distance 1
preserve the Hamming distance between all effects.  Furthermore,
the set of Hamming distance preserving transformations can be
expressed as combinations of permutations of subsystems and local
permutations (see for example Theorem~3.54 of~\cite{HammingDistance}).
We give a proof of this for completeness.

\begin{lemma} \label{lem:8} Let $\A$ be a finite alphabet, with $\A^N$
  the set of length-$N$ words.  Further, let $G_\Pi\cong S_N$ be the
  set of permutations of letters and $G_L \cong S_{|A|}^N$ be the
  group of local transformations of $\A^N$, which act independently at
  each position.
	
  Assume that $T^\dagger:\A^N\to\A^N$ is
  invertible. If $T^\dagger$ has the property that for all
  $s,t\in\A^N$, $d_H(s,t)=1\implies d_H(T^\dagger s,T^\dagger
  t)=1$, then $T^\dagger$ is a composition of operations from $G_\Pi$
  and $G_L$.
\end{lemma}

\begin{proof}
  Choose an arbitrary set $a_i\in A$, for $i=1,\dots,N$. Set
  $s=(a_1,\dots, a_N)$.  Left-multiplying $T^\dagger$ by a local
  operation if necessary, we may assume that $T^\dagger(s)=s$.

  For $i\in1\dots N$, consider the set, $L_i$ of words of the form
  \begin{equation*}
    L_i =(a_1,\dots, a_{i-1}, A_i, a_{i+1},\dots, a_N),
  \end{equation*}
  i.e.\ $s$ and strings that differ from $s$ only at position
  $i$. Because the elements of $L_i$ all have mutual Hamming distance
  equal to one, there must be a function $\pi$ such that
  $T^\dagger(L_i)=L_{\pi(i)}$. Since $T^\dagger$ is invertible, $\pi$ is a
  permutation, which may be thought of as an element of $G_\Pi$.
  Because $(\pi^{-1}) T^\dagger(L_i)=L_i$ for all $i$, there is no loss of
  generality in assuming that $T^\dagger$ takes $L_i$ to itself.  Employing
  yet another local operation if necessary, we may even assume that
  $T$ acts like identity on all elements of $L_i$, and hence on all
  strings with Hamming distance 1 to $s$.

  Define the \emph{weight} of an element $t\in \A^N$ to be
  $\operatorname{wt}(t)=d_H(t,s)$.  What we have shown so far amounts
  to the fact that $T^\dagger$ fixes all words of weight zero and one.
  Next, we prove by induction that $T^\dagger$ fixes the words of any
  weight $w$ (and hence all of $\A^N)$.

  Assume the claim has been established for weights up to $w-1$. If
  $t$ has weight $w>1$, it is uniquely specified by the $w$ words
  $r_i$ which have weight $\operatorname{wt}(r_i)=w-1$ and Hamming
  distance $d_H(r_i,t)=1$ to $t$ (in fact, any two words $r_i, r_j$
  from this set are sufficient to specify $t$). But since the weights
  of all of the $r_i$, and $d_H(r_i,t)$ are preserved by $T^\dagger$
  by the induction hypothesis, $T^\dagger$ must fix $t$.
\end{proof}

It follows that all transformations in boxworld can be formed by
composing subsystem permutations and local permutations.  However, the
set of allowed local permutations is further restricted (Hamming
distance preservation is a necessary but not sufficient condition on
$T^\dagger$):
\begin{lemma}\label{lem:9}
  The only local reversible operations allowed in boxworld are relabellings of
  measurements and their outcomes (separately for each measurement).
  Furthermore, all possible local relabellings are allowed transformations,
  regardless of the total number of subsystems, $N$.
\end{lemma}
\begin{proof}
  Recall~\eqref{eq:decomp} and note that these are the only
  combinations of extremal effects that sum to $\eins$:
  otherwise the identity $\eins$ could be decomposed into a sum
  involving two effects $X_m(k)$ and $X_{m'}(k')$ with $m\neq m'$, so
  there would be a state $s$ with $\langle X_m(k),s\rangle=\langle
  X_{m'}(k'),s\rangle=1$, for which $\langle \eins,s\rangle\geq 2$, a
  contradiction.
  
  Consider a measurement $m$, then
  \[
     \eins=T^\dagger \eins=T^\dagger \sum_k X_m(k)=\sum_k T^\dagger X_m(k),
  \]
  so each member of $\{T^\dagger X_m(k):0\leq k \leq K(m)-1\}$ must correspond
  to the same measurement. Hence all local reversible adjoint transformations
  permute the measurements, and, for each measurement separately, permute the
  outcomes.
  
  To see that all permutations are allowed if the number of outcomes
  $K$ is the same for every measurement $m$, note that the
  representation of $X_m(k)$ as in~(\ref{eq:repXm}) on $(\R^M\otimes
  \R^{K-1})\oplus\R$ permits that all those permutations are
  implemented as allowed linear (hence orthogonal) transformations:
  relabelling the measurements corresponds to permuting the standard
  unit vectors $e_m$ of $\R^M$ (constituting an $M$-dimensional
  irreducible representation of the symmetric group $S_M$), while
  relabelling the outcomes corresponds to symmetry transformations of
  the $(K-1)$-simplex in $\R^{K-1}$ with $K$ vertices $w_k$ (a
  $(K-1)$-dimensional irreducible representation of $S_K$).
  
  In the case that the number of outcomes $K=K(m)$ depends on the
  measurement $m$, the vector space carrying the local effects will
  analogously be $\bigoplus_{m=0}^{M-1} \R^{K(m)-1} \oplus\R$.  This
  allows us to represent the permutation of outcomes linearly, as
  before, while the permutations of measurements $m$ and $m'$ with
  $K(m)=K(m')$ correspond to permutations of direct summands.
  
  We have thus proven that every local relabelling transformation
  $T^\dagger$ is an allowed transformation in boxworld. It remains to
  show that $T^\dagger$ is allowed if the single system is coupled to
  others (i.e.\ that $T^\dagger\otimes\Id$ is an allowed
  transformation). (The analogue in quantum theory is to prove
  complete positivity).  This follows from the fact that local
  relabellings preserve the no-signalling, positivity and
  normalization constraints, such that $T\otimes\Id$ maps the
  no-signalling polytope (that is, the state space) onto itself.
\end{proof}

In the final part of the appendix, we prove that a state $s$ is a pure
product state if and only if all the probabilities $\langle
A,s\rangle$ are either $0$ or $1$ with respect to extremal effects,
$A$ (in the most general case that the number of measurements and
outcomes varies from site to site). This has been used in the proof of
Theorem 2.\\

\begin{lemma}
\label{lem:extremal01}
Let $s\in S^{(N)}$ be a normalized state on $N$ arbitrary boxworld
subsystems (some of which may be classical).  Then, $s$ is a pure
product state if and only if $\langle A,s\rangle\in\{0,1\}$ for all
extremal effects $A=A_1\otimes\ldots\otimes A_N\in\P^{(N)}$.
\end{lemma}
\begin{proof}
  If $s=s_1\otimes\ldots\otimes s_N$ is a product of pure states, then
  $\langle A_i,s_i\rangle\in\{0,1\}$ for every $i$, such that $\langle
  A,s\rangle$ is either $0$ or $1$. It remains to prove the converse.
  Suppose that $s$ is any state with $\langle A,s\rangle\in\{0,1\}$
  for all $A\in\P^{(N)}$. The idea is to construct a pure product
  state $\tilde s$ with $\langle A,s\rangle=\langle A,\tilde s\rangle$
  for all $A$, which proves that $s=\tilde s$.  To this end, note that
  the decomposition of the identity given in~\eqref{eq:id1} has the
  following generalization: if the $m_i$ are arbitrary local
  measurement devices, then
\begin{equation}\label{eq:id2}
   \eins^{(N)}=\sum_{k_1,\ldots,k_N} X_{m_1}(k_1)\otimes X_{m_2}(k_2)\otimes \ldots \otimes X_{m_N}(k_N),
\end{equation}
where the sum is over all outcomes (the number of outcomes may depend
on the measurement). It follows that
\[
   1=\sum_{k_1,\ldots,k_N}\langle X_{m_1}(k_1)\otimes X_{m_2}(k_2)\otimes \ldots \otimes X_{m_N}(k_N),s\rangle,
\]
so exactly one of the addends must be $1$, while all others are
$0$. Hence, to every string $\mathbf{m}=(m_1,m_2,\ldots,m_N)$
describing local choices of measurements, there is a unique string of
corresponding outcomes $\mathbf{k}(\mathbf{m}) =(k_1,k_2,\ldots,k_N)$
such that $\langle X_{m_1}(k_1)\otimes\ldots\otimes
X_{m_N}(k_N),s\rangle=1$, while the inner product is $0$ for all other
outcome combinations. It remains to show that each $k_i$ only depends
on $m_i$; in this case, we can construct a state $\tilde s=\tilde
s_1\otimes\ldots\otimes\tilde s_N$, generating the same probabilities
as $s$, factor by factor. So suppose there were strings $\mathbf{m}$
and $\mathbf{\tilde m}$ with $m_i=\tilde m_i$, but $k_i\neq \tilde
k_i$, (where $\tilde{\bf k}:={\bf k}(\tilde{\bf m})$) then
\begin{eqnarray*}
   1&=& \langle \eins\otimes\ldots\otimes\eins\otimes \sum_{k'}X_{m_i}(k')\otimes\eins\otimes\ldots\otimes \eins,s\rangle\\
   &\geq& \langle \eins\otimes\ldots\otimes\eins\otimes X_{m_i}(k_i)\otimes\eins\otimes\ldots\otimes\eins,s\rangle \\
   && + \langle \eins\otimes\ldots\otimes\eins\otimes X_{\tilde m_i}(\tilde k_i)\otimes\eins\otimes\ldots\otimes\eins,s\rangle \\
   &\geq& \langle X_{m_1}(k_1)\otimes\ldots\otimes X_{m_N}(k_N),s\rangle \\
   && + \langle X_{\tilde m_1}(\tilde k_1)\otimes\ldots \otimes X_{\tilde m_N}(\tilde k_N),s\rangle \\
   &=& 2
\end{eqnarray*}
where the last inequality follows by applying the decomposition of the
identity~\eqref{eq:id2} to $N-1$ factors.  This is a contradiction.
\end{proof}

\end{document}